\def\ipad{x}
\ifdefined\ipad{}
\documentclass[letter,english,11pt]{article}
\usepackage{amsfonts}
\usepackage{amsmath}
\usepackage{amsthm}
\usepackage{amssymb}
\usepackage{geometry}
\usepackage[english]{babel}
\else
\documentclass[sigplan,screen,review]{acmart}
\fi
\usepackage[inline]{enumitem} 
\usepackage{mathtools} 
\usepackage{xspace} 
\usepackage{bm} 
\usepackage{tikz} 
\usetikzlibrary{shapes,arrows,positioning,patterns}
\usepackage{mathpazo}
\usepackage{palatino}

\usepackage{hyperref}

\usepackage{hyperref} 
\usepackage[capitalise]{cleveref} 
\hypersetup{
    colorlinks=true,
    citecolor=blue,
    linkcolor=blue,
    filecolor=blue,
    urlcolor=blue,
}

\usepackage{framed}
\usepackage{mdframed}
\usepackage{float}
\usepackage{xcolor}

\newenvironment{proof*}[1][\proofname]{
  
  \begin{proof}[#1]}{\end{proof}}

\newlength{\leftbarwidth}
\newlength{\leftbarsep}
\newtheorem{innercustomthm}{Theorem}
\newenvironment{customthm}[1]
  {\renewcommand\theinnercustomthm{#1}\innercustomthm}
  {\endinnercustomthm}

\theoremstyle{plain}
\newtheorem{theorem}{Theorem}[section]

\newtheorem{question}[theorem]{Question}
\newtheorem{corollary}[theorem]{Corollary}
\newtheorem{lemma}[theorem]{Lemma}
\newtheorem{definition}[theorem]{Definition}
\newtheorem{claim}[theorem]{Claim}

\newcommand{\str}[1]{{\mathbb {#1}}}
\newcommand{\from}{\colon}

\newcommand{\set}[1]{\{#1\}}
\newcommand{\setof}[2]{\set{#1\mid#2}}
\newcommand{\wh}{\widehat}
\def\phi{\varphi}
\def\cal{\mathcal}
\def\N{\mathbb N}
\def\R{\mathbf R}
\def\epsilon{\varepsilon}

\renewcommand{\subset}{\subseteq}
\renewcommand{\setminus}{-}
\renewcommand{\le}{\leqslant}

\newcommand{\dist}{\textup{dist}}

\newcommand\lcw{\textup{lcw}}
\newcommand\cw{{\rm cw}}

\newcommand{\CC}{\cal C}
\newcommand{\DD}{\cal D}
\newcommand{\trans}[1]{\mathsf{#1}}
\newcommand{\interp}[1]{\mathsf{#1}}

\newcommand{\FO}{\text{\rm{FO}}\xspace}
\newcommand{\MSO}{\text{\rm{MSO}}\xspace}

\newcommand{\FPT}{\text{\rm{FPT}}\xspace}

\newcommand{\Oof}{\mathcal O}

\newcommand{\tup}{\bar}

\newcommand{\CCC}{\mathcal{C}}

\newcommand{\LLL}{\mathcal{L}}

\newcommand{\sth}{\mathrel : }

\renewcommand{\le}{\leqslant}
\renewcommand{\ge}{\geqslant}
\renewcommand{\phi}{\varphi}

\renewcommand{\subset}{\subseteq}
\newcommand{\ind}[1][]{%
  \mathrel{
    \mathop{
      \vcenter{
        \hbox{\oalign{\noalign{\kern-.3ex}\hfil$\vert$\hfil\cr
              \noalign{\kern-1.5ex}
              $\vert$\cr\noalign{\kern-.3ex}}}
      }
    }\displaylimits_{#1}
  }
}

\newcommand{\nind}[1][]{%
  \mathrel{
    \mathop{
      \vcenter{
        \hbox{\oalign{\noalign{\kern-.3ex}\hfil$\nmid$\hfil\cr
              \noalign{\kern-1.5ex}
              $\vert$\cr\noalign{\kern-.3ex}              
              }}
      }
    }\displaylimits_{#1}
  }
}

\ifdefined\ipad{}

\author{\'{E}douard Bonnet\thanks{Univ Lyon, CNRS, ENS de Lyon, UCBL 1, LIP UMR5668, France, {edouard.bonnet@ens-lyon.fr}}
\and Jan Dreier\thanks{TU Wien, Austria, {dreier@ac.tuwien.ac.at}}
\and Jakub Gajarsk\'y\thanks{University of Warsaw, Poland, {gajarsky@mimuw.edu.pl}}
\and {Stephan Kreutzer}\thanks{{TU Berlin, Germany}, {kreutzer@tu-berlin.de}{}{}}
\and {Nikolas M\"ahlmann}\thanks{{University of Bremen, Germany}, {maehlmann@uni-bremen.de}{}{}}
\and {Pierre Simon}\thanks{{University of Berkeley, USA}, {pierre.simon@berkeley.edu}{}{}}
\and {Szymon Toru\'nczyk}\thanks{{University of Warsaw, Poland}, {szymtor@mimuw.edu.pl}}}
\title{Model Checking on Interpretations of Classes\\of Bounded Local Cliquewidth
\footnote{This research was initiated at the Dagstuhl workshop \emph{Sparsity in Algorithms, Combinatorics and Logic} (September 2021). We wish to thank the organizers and other participants.
E.B. was supported by the ANR projects TWIN-WIDTH (ANR-21-CE48-0014) and Digraphs (ANR-19-CE48-0013).
J.G. and S.T. were supported by the project BOBR that is funded from the European Research Council (ERC) under the European Union’s Horizon 2020 research and innovation programme (grant agreements No. 683080 and 948057, respectively).
N.M. was supported by the German Research Foundation (DFG) with grant greement No. 444419611.
}
}

\else
\author{\'{E}douard Bonnet}{Univ Lyon, CNRS, ENS de Lyon, Université Claude Bernard Lyon 1, LIP UMR5668, France \and \url{http://perso.ens-lyon.fr/edouard.bonnet/}}{edouard.bonnet@ens-lyon.fr}{https://orcid.org/0000-0002-1653-5822}{}
\author{Jan Dreier}{TU Wien, Austria \and \url{https://www.ac.tuwien.ac.at/people/dreier/}}{dreier@ac.tuwien.ac.at}{https://orcid.org/0000-0002-2662-5303}{}
\author{Jakub Gajarsk\'y}{University of Warsaw, Poland \and \url{https://sites.google.com/view/jakubgajarsky/}}{gajarsky@mimuw.edu.pl}{}{}
\author{Stephan Kreutzer}{Logic and Semantics, TU Berlin, Germany}{kreutzer@tu-berlin.de}{}{}
\author{Nikolas M\"ahlmann}{University of Bremen, Germany}{maehlmann@uni-bremen.de}{}{}
\author{Pierre Simon}{University of Berkeley, USA}{pierre.simon@berkeley.edu}{}{}
\author{Szymon Toru\'nczyk}{University of Warsaw, Poland \and \url{https://www.mimuw.edu.pl/~szymtor/}}{szymtor@mimuw.edu.pl}{https://orcid.org/0000-0002-1130-9033}{}

\title{Model Checking on Interpretations of Classes of Bounded Local Cliquewidth}
\fi

\ifdefined\ipad{}
\else

\titlerunning{Model Checking on Interpretations of Classes of Bounded Local Cliquewidth}

\authorrunning{\'E. Bonnet, J. Dreier, J. Gajarsk\'y, S. Kreutzer, N. M\"ahlmann, S. Toru\'nczyk}

\Copyright{\'Edouard Bonnet, Jan Dreier, Jakub Gajarsk\'y, Stephan Kreutzer, Nikolas M\"ahlmann, Szymon Toru\'nczyk}



\category{}

\relatedversion{}

\supplement{}

\funding{\'E. B. was supported by the ANR projects TWIN-WIDTH (ANR-21-CE48-0014) and Digraphs (ANR-19-CE48-0013).}

\funding{(Optional) general funding statement \dots}

\acknowledgements{This research was initiated at the Dagstuhl workshop \emph{Sparsity in Algorithms, Combinatorics and Logic} (September 2021). We wish to thank the organizers and other participants.}

\nolinenumbers 

\hideLIPIcs  
\fi

\begin{document}

\maketitle

\begin{abstract}
We present a fixed-parameter tractable algorithm for first-order model checking on interpretations of graph classes with bounded local cliquewidth.
Notably, this includes interpretations of planar graphs, and more generally, of classes of bounded genus.
To obtain this result we develop a new tool which works in a very general setting of dependent classes and which we believe can be an important ingredient in achieving similar results in the future.



\end{abstract}

\maketitle

\section{Introduction}

Algorithmic meta-theorems aim to explain the tractability of entire families of problems that can be specified in some logic.
The prime example is Courcelle's theorem \cite{courcelle90}, stating that every problem expressible in monadic second-order logic (\MSO) can be solved in linear time on every class of graphs with bounded treewidth.
In this paper, we follow a long line of research concerned with algorithmic meta-theorems for first-order logic (\FO), on restricted classes of graphs.
The central problem here is the first-order model checking problem, where one should decide whether a given \FO sentence $\phi$ holds in a given graph $G$.
A naive algorithm solves this problem in time $\Oof(|G|^{|\phi|})$ whereas no algorithm can solve it in time $|G|^{o(|\phi|)}$ in general, unless \textsc{SAT} admits a subexponential-time algorithm.
The main goal of this line of research is to identify classes of graphs for which the problem is
\emph{fixed-parameter tractable} (FPT), i.e., solvable in time $f(\phi)\cdot |G|^c$, for some constant $c$ and computable function $f\from\mathbb N\to\mathbb N$.
Henceforth we call such classes \emph{tractable}.
Courcelle's theorem gives such an algorithm even for the more powerful logic \MSO, on all classes of bounded treewidth.

The first result of this kind for \FO, proven by Seese~\cite{seese96},
states that \FO~model checking is \FPT on every class of graphs with bounded maximum degree.
This result is also the first application of the \emph{locality method}, utilizing the locality of first-order logic, as formalized for example by Gaifman's locality theorem.
Gaifman's theorem implies in particular that for two vertices $u,v$ of a graph $G$ (that are sufficiently far apart), whether or not $u$ and $v$ satisfy a fixed formula $\phi(x,y)$ can be determined by looking only 
at neighborhoods of bounded radius around $u$ and around $v$ in $G$.
The locality method was extended by Frick and Grohe~\cite{frickg01} who showed that 
if there is an \FPT algorithm for all classes $\CC$ satisfying a certain property $\cal P$ (where the exponent in the run time of the algorithm is the same for all $\CC\in\cal P$),
then this immediately implies the existence of such an \FPT algorithm for all classes $\CC$ that \emph{locally} have property $\cal P$.
A class $\CC$ has \emph{locally} property $\mathcal P$ if 
for every fixed radius $r$, the class of all $r$-balls of graphs from $\CC$ has property $\mathcal P$.
For example, a class $\CC$ has  locally bounded treewidth if there is a function $f\from\N\to\N$ such that 
for every $G\in \CC$ and vertex $v\in V(G)$, the subgraph of $G$ induced by the $r$-ball around $v$ 
has treewidth at most $f(r)$. Such classes are also said to have \emph{bounded local treewidth}.
Planar graphs, graphs of bounded genus, and more generally, apex-minor-free graphs, have bounded local treewidth, so \FO~model checking is \FPT~on all those classes, by the observation of Frick and Grohe combined with the result of Courcelle.
The locality method was subsequently combined with the graph minor theory of Robertson and Seymour,
to capture all classes that exclude a minor~\cite{flumg01},
or more generally, classes that locally exclude a minor~\cite{dawargk07}.

A new paradigm, based on sparsity theory
developed by Ne\v set\v ril and Ossona de Mendez, has allowed to obtain 
further, more general tractability results.
Dvo\v{r}\'{a}k, Kr\'{a}l and Tho\-mas
showed that \FO~model checking is {\FPT} for every class with \emph{bounded expansion}~\cite{dvorakkt10}.
And finally, Grohe, Kreutzer and Siebertz 
showed that the same holds for every \emph{nowhere dense} graph class~\cite{gks17}. 
Those include all the classes mentioned above.
See also \Cref{fig:sparse_inclusion_diagramm} for the relationship between these classes.

\begin{figure}[ht]
\centering
    \includegraphics[scale=0.75]{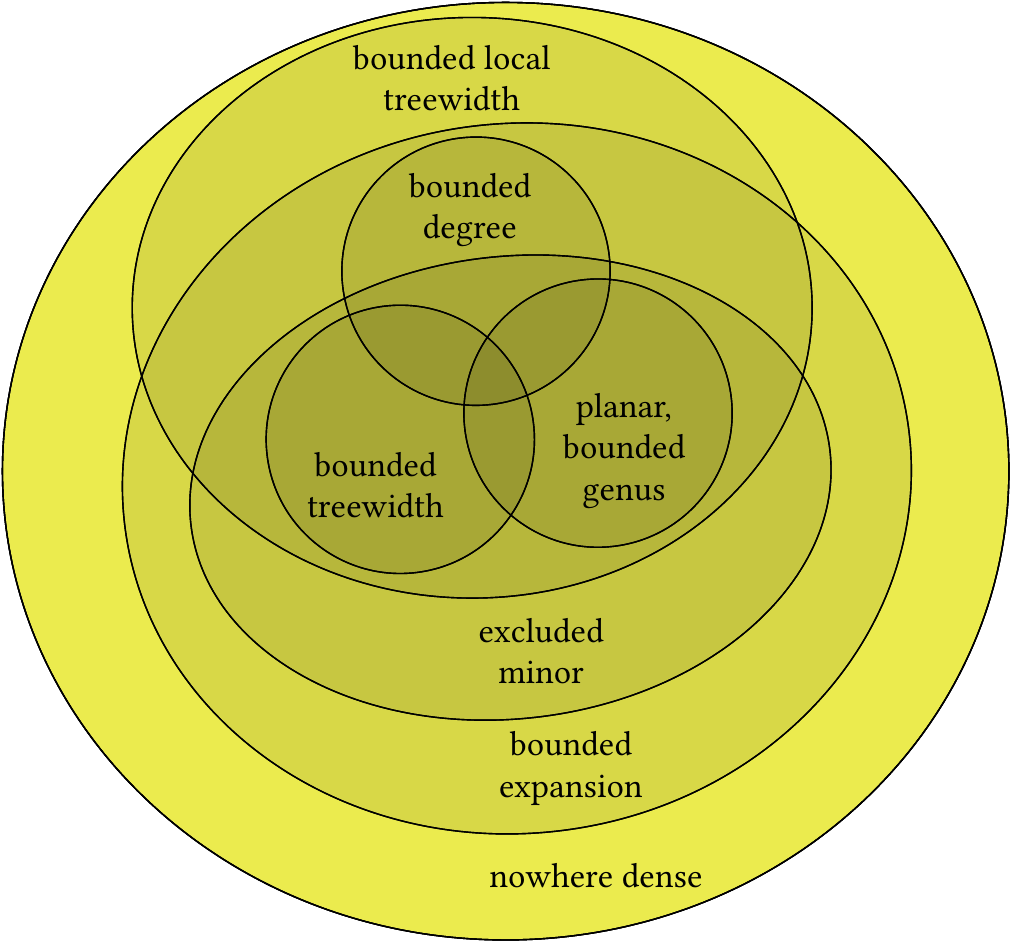}
    \caption{Inclusion diagram of selected monotone graph classes
        with FPT \FO model checking.}
    \label{fig:sparse_inclusion_diagramm}
\end{figure}

All graph classes we discussed so far are \emph{monotone}, i.e., closed under removing vertices and edges.
For sparse graph classes, monotonicity appears to be a reasonable assumption: after all, removing edges from a sparse graph should only make it even sparser.

For monotone graph classes, the aforementioned results are beautifully complemented by matching lower bounds.
MSO$_2$ model checking is not \FPT on monotone graph classes whose treewidth is at least polylogarithmic with respect to the number of vertices~\cite{kreutzer2010lower,GANIAN2014180}
and FO model checking is not \FPT on monotone graph classes that are not nowhere dense~\cite{dvorakkt10,kreutzer11}.
Thus, the aforementioned results yield a complete characterization of the monotone graph classes admitting \FPT model checking of \FO,
and an almost complete characterization of the monotone graph classes admitting \FPT model checking of $\MSO_2$.

However, this is far from the complete picture, as this says nothing about the tractability for \emph{dense} graph classes.
Simple examples of graph classes that are not monotone, but admit efficient FO model checking are the class of complete graphs, or more generally, the class of edge complements of graphs from a fixed nowhere dense class.
Those are not contained in any tractable monotone graph class, as every monotone graph class that contains cliques of unbounded size also contains all graphs.
Thus, to make further progress, we need a paradigm shift towards considering non-monotone graph classes and width measures.

\paragraph*{Dense Graph Classes.}
A graph class $\CC$ is \emph{hereditary} if $\CC$ is closed under taking \emph{induced} subgraphs, that is, under removing vertices.
Since we do not assume closure under edge removal, hereditary classes are well suited to capture dense graph classes.
After the question for monotone classes has been settled, the major next goal
is to characterize hereditary graph classes for which \FO model checking is FPT.
This is again inspired by results for \MSO model checking:
The result by Courcelle, Makowsky, and Rotics~\cite{cw}, combined with the result of Oum and Seymour~\cite{cw_rw}, shows that \MSO model checking is FPT on classes of \emph{bounded cliquewidth}. Cliquewidth is a generalization of the notion of treewidth to dense graphs. In particular, it is preserved by taking edge complements. 

Applying again the locality argument to classes of bounded cliquewidth yields the following result, originating in the work of Frick and Grohe~\cite{frickg01}.
Say that a class $\CC$ has \emph{bounded local cliquewidth} if there is a function $f\from\N\to\N$ such that for every number $r\in\N$, graph $G\in\CC$, and vertex $v\in V(G)$, the subgraph of $G$ induced by the $r$-ball around $v$ has cliquewidth at most $f(r)$.

\begin{theorem}\label{thm:loc-cw}
  Let $\CC$ be a class with bounded local cliquewidth. 
  Then \FO model checking is fixed-parameter tractable on $\CC$.
\end{theorem}
%
Currently, 
classes of bounded local cliquewidth are one of a few dense families
for which \FO model checking is known to be fixed-parameter tractable. 
However, there are many other graph classes that are conjectured to be tractable (for a more detailed discussion, see \Cref{fig:dense_inclusion_diagramm} and \cref{sec:discussion}).
Those include, in particular, classes that can be obtained from tractable classes,
using \FO formulas, as we make precise now.

\paragraph*{Interpretations.}
Let $\Sigma$ and $\Gamma$ be two signatures, where $\Gamma$ is relational.
A \emph{simple interpretation}
$\interp I\from \Sigma\to\Gamma$ (here, \emph{interpretation} for short) is specified by a \emph{domain} formula $\delta(x)$,
and one formula  $\phi_R(x_1,\ldots,x_k)$ for each relation symbol $R\in \Gamma$ of arity $k$, where all those formulas are in the signature $\Sigma$.
For a given $\Sigma$-structure $\str A$, the interpretation $\interp I$ 
outputs the structure $\interp I(\str A)$ whose domain is the set $\delta(\str A):=\setof{ a\in \str A} {\str A\models \delta(a)}$,
and in which the interpretation of each relation $R\in\Gamma$
of arity $k$ consists of those tuples $(a_1,\ldots,a_k)\in \delta(\str A)^k$ that satisfy $\str A\models \phi_R(a_1,\ldots,a_k)$.
Usually we will be working with interpretations that map graphs with expanded signatures to uncolored, undirected graphs, having a single binary relation $E$.
In this case,
we write $\interp I_{\phi,\delta}$ for the interpretation consisting of an irreflexive, symmetric formula $\phi(x,y)$ interpreting the edge relation $E$ and a domain formula $\delta(x)$.
If $\delta(x)$ is equal to  $x=x$, we will just write $\interp I_\phi$ instead.
For example, the interpretation $\interp I_\phi$ with $\phi(x,y)=\neg E(x,y)$  maps a given graph 
to its edge complement,
and the interpretation $\interp I_\psi$ with with $\psi(x,y)=E(x,y) \lor \exists z.E(x,z) \land E(z,y)$ maps a given graph to its square.

The notion of an interpretation lifts to classes of structures, for which we denote with
$\interp I(\mathcal{C}) := \setof{\interp I(G)}{G \in
\mathcal{C}}$ the result of applying the interpretation $\interp I$ to the class $\CC$.
Say that a class of structures $\CC$ \emph{interprets} a class of structures~$\DD$,
or that $\DD$ \emph{interprets in} $\CC$,
if there is an interpretation~$\interp I$ such that 
$\DD\subset \interp I(\CC)$.
Note that this notion depends on the chosen underlying logic, which will be either \FO or \MSO in our discussion.
We may write \emph{$\mathcal L$-interpretation} for interpretation when the underlying logic is $\mathcal L \in \{\FO,\MSO\}$. 
This  yields a transitive relation: if $\CC$ interprets $\DD$ and $\DD$ interprets $\cal E$, then $\CC$ interprets $\cal E$.

A class $\CC$ of graphs has bounded cliquewidth 
if and only if the class of trees $\MSO$-interprets $\CC$~\cite[Proposition 27]{bc06}.
In particular, bounded cliquewidth is preserved by \MSO interpretations. Moreover, 
we may view interpretations as a tool to extend model checking results from sparse to dense graph classes.
This invites the question, originally asked in~\cite{gajarsky_2020}, whether a similar statement holds for first-order logic.

\begin{question}\cite{gajarsky_2020}\label{qstn:transductions_preserve_tractability}
Let $\CC$ be a class admitting an {\FPT} algorithm for {\FO} model checking,
and $\DD$ be a class that {\FO}-interprets in~$\CC$. Does there exist 
an {\FPT} algorithm for {\FO} model checking on $\DD$?
\end{question}
The intuition underlying this question is that if a graph class $\CC$ is sufficiently well-behaved,
 then a fixed formula $\phi(x,y)$ should not be able to define 
 complicated graphs in graphs from $\CC$.

Thus, in particular, 
Question~\ref{qstn:transductions_preserve_tractability} 
suggests the existence of an {\FPT} algorithm for {\FO} model checking for any  class $\DD$ that interprets in some nowhere dense class $\CC$.
How could such an algorithm look like?
To unravel this question, fix an interpretation $\interp I$ such that $\DD\subset \interp I(\CC)$, where $\CC$ is the class of $k$-colored graphs from $\CC$.
Given a graph $G\in\DD$ and a first-order formula $\phi$ that we want to evaluate on $G$,
a possible strategy is to try to ``reverse the interpretation'' and compute
a graph $G'\in \CC$ such that $\interp I(G')=G$.
This process then yields a formula $\phi'$ such that $G \models \phi$ if and
only if $G' \models \phi'$.
Since $G'$ comes from a nowhere dense class, one can then evaluate in FPT time whether $G' \models \phi'$. 
However, reversing an interpretation seems to be a difficult task\footnote{For instance, it is NP-complete \cite{mms94} to decide whether 
a given graph is a square of some graph.}.
In this approach, we do not necessarily need to revert the interpretation $\interp I$ as described above --  there may be some other nowhere dense class $\CC'$ and interpretation $\interp I'$ that is easier to revert, such that $\DD\subset \interp I'(\CC')$. 

So far, only for classes which interpret in bounded degree classes the method outlined above has been applied successfully~\cite{gajarsky_2020}.
For the more general classes interpreting in bounded expansion classes, an FPT \FO model checking now boils down to efficiently computing so-called \emph{low shrubdepth covers}~\cite{GajarskyKNMPST20}, or Lacon or shrub decompositions~\cite{Dreier21}.

\paragraph*{Main result.}
In this paper, we extend the result of \cite{gajarsky_2020} significantly
by proving that {\FO} model checking is FPT 
for every class that interprets in a class with bounded local cliquewidth.
\begin{theorem}[Main result]\label{thm:mc_bounded_local_cliquewidth}
  Let $\CC$ be a graph class that interprets in a class of graphs  with bounded local cliquewidth.
  Then \FO model checking is fixed-parameter tractable on $\CC$:
  there exists a function $f$ and a constant $c$ such that for every first-order sentence $\phi$ and graph $G \in \cal C$
  one can decide in time $f(|\phi|)\cdot n^c$ whether $G \models \phi$.
  \end{theorem}

\begin{figure}[t]
\centering
  \includegraphics[scale=0.75]{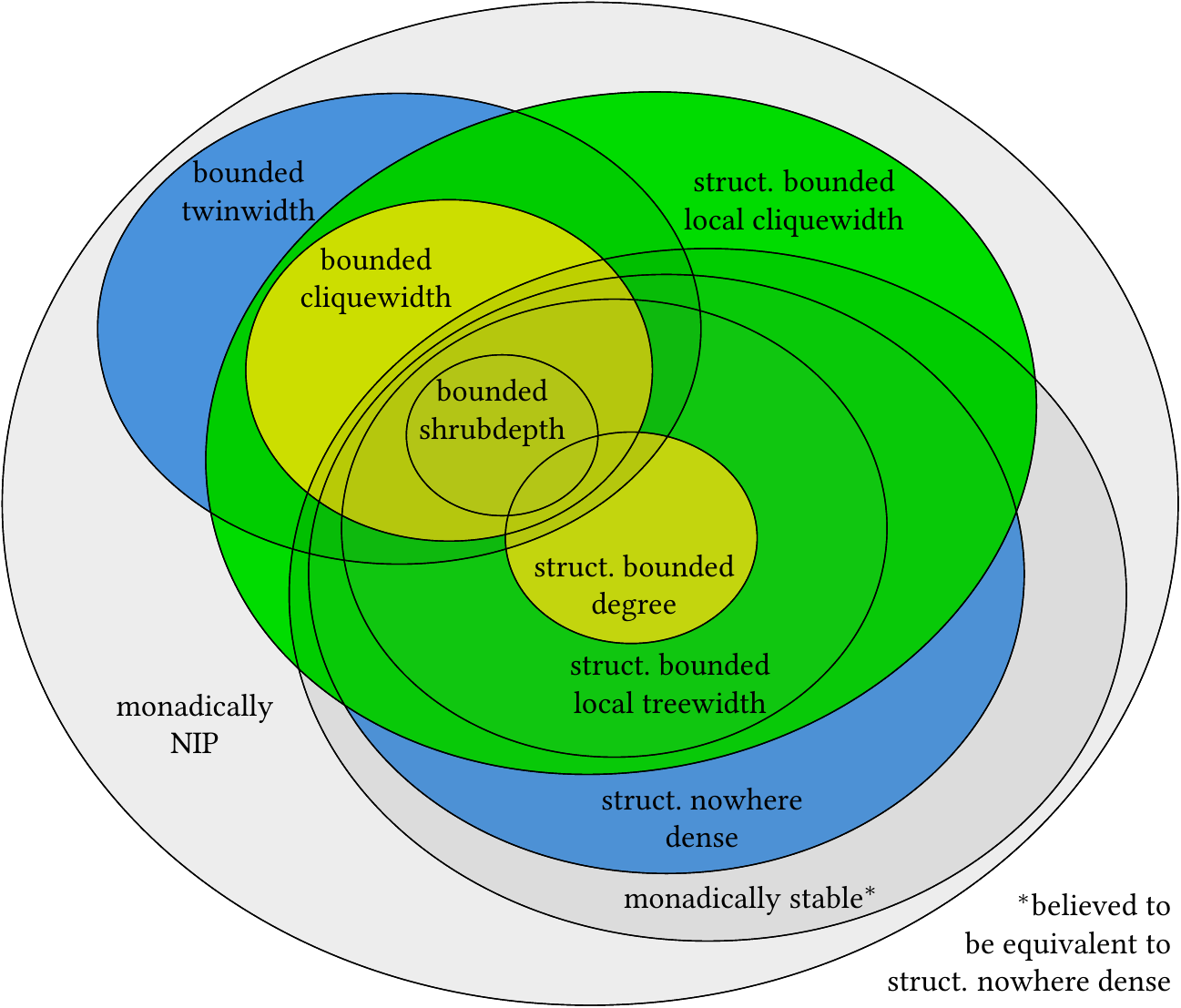}
  \caption{Inclusion diagram of selected transduction ideals, that is,
  properties of graph classes that are closed under transductions (that is, under interpretations of colorings of the graphs from the class). 
  Yellow transduction ideals were previously known to admit an FPT \FO model checking algorithm.
  Green transduction ideals admit an FPT \FO model checking algorithm, as presented in this paper.
  Blue ideals admit an FPT \FO model checking algorithm, assuming an appropriate decomposition is given as part of the input.
  Uncolored means unknown.
  The relevant notions are discussed in \Cref{sec:discussion}.
  }
  \label{fig:dense_inclusion_diagramm}
\end{figure}

Thus, we make progress towards answering \cref{qstn:transductions_preserve_tractability}, by answering it positively in the case of interpretations of classes with bounded local cliquewidth.
See \Cref{fig:dense_inclusion_diagramm} for an overview on how our result relates to previous results.
We remark that besides being much more general, our proof is also much simpler than 
the proof in \cite{gajarsky_2020}.
As we explain in the proof outline below, our main lemma applies to much more general classes than just classes of bounded local cliquewidth -- namely to all \emph{NIP classes} -- yielding a more general theorem than \cref{thm:mc_bounded_local_cliquewidth}
(see \cref{thm:local_iff_global}).
We proceed with a proof outline in \Cref{sec:outline},
followed by the actual proofs and then an extended 
discussion in \Cref{sec:discussion}, comparing our results to existing results.

\section{Proof outline}\label{sec:outline}
In this section, we sketch the proof of \cref{thm:mc_bounded_local_cliquewidth}.
This proof outline is not complete,
and for simplicity of the description assumes interpretations 
in which the domain formula $\delta(x)$ holds for all $x$.
For a complete proof 
see \cref{sec:main-lemma,sec:main-proof}.

We first describe a possible proof strategy for proving \cref{thm:mc_bounded_local_cliquewidth},
outlined in \cite{gajarsky2020differential}, in order to isolate the main obstacle.
The following lemma is an immediate consequence of Gaifman's locality theorem~\cite{gaifman82}.
By $\dist^G(u,v)$ denote the distance between two vertices $u$ and $v$ in a graph~$G$.
\begin{lemma}\label{lem:gaifman-coloring}
  Let $\phi(x,y)$ be an \FO formula. Then there are numbers $r,t\in \N$ 
  such that every graph $G$ can be vertex-colored using $t$ colors 
  in such a way that for any two vertices $u,v\in V(G)$ with $\dist^G(u,v)>r$,
  whether or not $\phi(u,v)$ holds depends only on the color of $u$ and the color of~$v$.
\end{lemma}

Rephrasing, the conclusion of \cref{lem:gaifman-coloring}
says that there is a formula $\alpha(x,y)$, which is a Boolean combination 
of checks of the colors of $x$ and $y$,
and is such that the formula $\psi(x,y):=\phi(x,y)\oplus \alpha(x,y)$ 
has \emph{range $\le r$}, that is, 
for every graph $G$ and vertices $u,v\in V(G)$, 
if $\psi(u,v)$ holds then $\dist^G(u,v)\le r$. Here, $\oplus$ denotes the exclusive or.

This has the following consequence, observed in~\cite{gajarsky2020differential,nesetril2021structural}.
If $G$ is a graph and $X,Y\subset V(G)$ are sets of vertices of $G$,
then doing a \emph{flip} between $X$ and $Y$ yields a new graph where the adjacency of all pairs $x\in X$ and $y\in Y$ is inverted:
adjacent pairs become non-adjacent, and vice-versa.
\begin{corollary}[\cite{gajarsky2020differential,nesetril2021structural}]
  \label{flip-cor}
  For every formula $\phi(x,y)$ there are $r,t\in\N$ and a formula
  $\psi(x,y)$ of range $\le r$ such that for every graph~$G$, the graph $\interp I_\phi(G)$ can be obtained from the graph $\interp I_\psi(G)$ by performing flips between $t$ pairs of sets.
\end{corollary}
To see this, perform a flip for every pair of color classes $C,D$ (as given by \cref{lem:gaifman-coloring}) such that $\phi(u,v)$ holds for some $u\in C$ and $v\in D$ 
with $\dist^G(u,v)>r$.
So the $t$ in \cref{flip-cor} is in fact at most the square of the $t$ obtained from \cref{lem:gaifman-coloring}.

Now, suppose we are given a class $\CC$ with bounded local cliquewidth
and an interpretation $\interp I_\phi$, for some \FO formula $\phi(x,y)$, 
and want to solve the model checking problem on the class $\interp I_\phi(\CC)$.
In this problem, we are given as input a graph of the form $\interp I_\phi(G)$, for some $G\in \CC$ which is unknown, and a sentence $\alpha$, and are to determine whether $\interp I_\phi(G)$ satisfies~$\alpha$.

Let $\psi$ be as in \cref{flip-cor}.
As $\CC$ has bounded local cliquewidth and $\psi(x,y)$ has range $\le r$, 
it is not difficult to prove that $\interp I_\psi(\CC)$ is again a class with bounded local cliquewidth (this relies on the fact that classes with bounded cliquewidth are closed under \FO-interpretations,
and is proved in \cref{lem:loc-bdcw}). Hence, \FO model checking can  efficiently be solved on the graph $\interp I_\psi(G)$ as given by \cref{flip-cor}.
To model check the sentence $\alpha$ on $\interp I_\phi(G)$
it is enough to model-check another sentence $\alpha'$ on the graph $\interp I_\psi(G)$ expanded 
with unary predicates marking the $t$ pairs of sets that need to be flipped 
to obtain $\interp I_\psi(G)$ from $\interp I_\phi(G)$. Here we use the fact that the same flips can be used to recover $\interp I_\phi(G)$ from $\interp I_\psi(G)$, and the flipping process can be simulated by $\alpha'$.

To summarize, to determine whether $\interp I_\phi(G)$ satisfies $\alpha$,
it suffices to determine whether $\interp I_\psi(G)$ (with additional unary predicates) satisfies $\alpha'$, and this can be done efficiently since $\interp I_\psi(\CC)$ has bounded local cliquewidth. And moreover $\interp I_\psi(G)$ can be obtained from $\interp I_\phi(G)$ by performing $t$ flips between pairs of sets. The problem with this approach is: how to determine the $t$ pairs of sets that need to be flipped in order to obtain $\interp I_\psi(G)$ from $\interp I_\phi(G)$?
Lemma~\ref{lem:gaifman-coloring} allows us to find those sets when $G$ is given, but not when $\interp I_\phi(G)$ is given.

Our main lemma overcomes this difficulty by proving a version of Corollary~\ref{flip-cor} in which the $t$ flips can be efficiently computed, given $\interp I_\phi(G)$.
Before we can state it, we will need the following fundamental notions originating from learning theory.

\paragraph*{VC-dimension and NIP classes.}
Say that a formula $\phi(x,y)$ has \emph{VC-dimension} at least $N$ on a structure $G$ if there exist elements $v_i$ for $i=1,\ldots,N$ and  $w_I$ for $I\subset \set{1,\ldots,N}$ 
such that $\phi(v_i,w_I)$ holds if and only if $i\in I$, for all $i=1,\ldots,N$ and $I\subset \set{1,\ldots,N}$.
See also \Cref{fig:VC}.
\begin{figure}
\centering
    \includegraphics{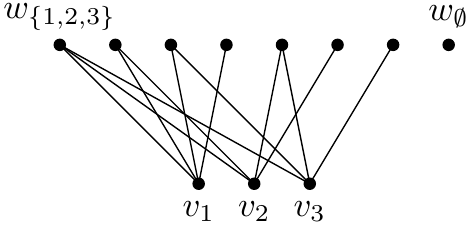}
    \caption{A graph for which the formula $\phi(x,y) = E(x,y)$ has VC-dimension at least three.}
    \label{fig:VC}
\end{figure}

A class $\CC$ of structures is \emph{NIP} (or \emph{dependent})
if for every first-order formula\footnote{In the original definition \cite{SHELAH1971271}, this condition is required for formulas $\phi(\tup x,\tup y)$, where $\tup x$ and $\tup y$ are tuples of variables. However, our proofs work with the weaker assumption.} $\phi(x,y)$ there is some constant $N$ 
such that the VC-dimension of $\phi$ on $G$ is less than $N$ for every ${G\in \CC}$.
Every class with bounded local cliquewidth is NIP~\cite{Grohe2003LearnabilityAD}. 
There are many other known NIP classes $\CC$,
such as all no\-where dense classes, and more generally, all \emph{monadically NIP classes}
(see \cref{sec:discussion}).

\paragraph*{Main lemma.}
We are now ready to state our main technical lemma, in a form that parallels \cref{lem:gaifman-coloring}.

\begin{lemma}[Main lemma]\label{lem:super-coloring-intro}
  Let $\CC$ be a class of graphs and 
  let $\phi(x,y)$ be an \FO formula that has bounded VC-dimension on~$\CC$. 
  Then there are numbers $s,r\in \N$ 
  such that for every $G\in\CC$ there is a 
set $S\subset V(G)$ of size at most $s$
such that for any two vertices $u,v\in V(G)$ with $\dist^G(u,v)>r$,
  whether or not $\phi(u,v)$ holds, depends only on $\phi(u,S)$ and $\phi(S,v)$.
\end{lemma}

Here, $\phi(u,S):=\setof{w\in S}{G\models \phi(u,w)}$, and $\phi(S,v)$ is defined symmetrically. In particular, if $\phi(x,y)\leftrightarrow \phi(y,x)$ holds, as is the case when considering formulas that define graphs, 
then we have that $\phi(u,S)=\phi(S,u)$. In what follows, we assume that $\phi(x,y)\leftrightarrow\phi(y,x)$ holds.

Given a set $S\subset V(G)$ define a coloring of $V(G)$ that colors a given $v\in V(G)$ with the set $\phi(v,S)\subset S$. This coloring then uses at most $2^s$ 
colors, and is moreover definable in a straightforward way in the graph $\interp I_\phi(G)$,
by looking at the adjacencies between a given vertex and the vertices in $S$.
The conclusion of the lemma says that for all $u$ and $v$ with ${\dist^G(u,v)>r}$,
whether or not $\phi(u,v)$ holds, depends only on the color of $u$ and the color of $v$, that is, there is some binary relation $R\subset 2^S\times 2^S$ such that 
$G\models\phi(u,v)$ if and only if the pair formed by the colors of $u$ and $v$ belongs to~$R$.
Hence, \cref{lem:super-coloring-intro} can be seen as a variant of \cref{lem:gaifman-coloring}, where the coloring can moreover be efficiently computed, given the graph $\interp I_\phi(G)$ and the set $S$.



\paragraph*{Main algorithm.}
Using \cref{lem:super-coloring-intro}, we can now solve the model checking problem 
on $\interp I_\phi(\CC)$, essentially in the way that was outlined above.
More precisely, the algorithm works as follows. Given a graph $\interp I_\phi(G)\in \interp I_\phi(\CC)$ 
and an \FO sentence~$\alpha$,
 in parallel for every set $S\subset V(G)$ with $|S|\le s$,
 and every binary relation $R\subset 2^S\times 2^S$,
 do the following.
 \begin{enumerate}
  \item Compute the coloring of $V(G)$ as described above, using $2^S$ colors.
  \item Compute the graph $\interp I_\psi(G)$ by performing flips 
  between any pair of color classes such that belongs to $R$.
  \item Check whether $\interp I_\psi(G)$ expanded with unary predicates marking the flipped sets,  
  satisfies $\alpha'$, where $\alpha'$ is the formula that first recovers the graph $\interp I_\phi(G)$ by undoing the  flips, and then tests whether $\interp I_\phi(G)$ satisfies $\alpha$.
\end{enumerate}
 Whenever one of the parallel executions terminates, terminate with the same  answer.

There is one technicality on which the proof of correctness of the above algorithm hinges. We do not know which of the parallel executions 
involves the ``correct'' set $S$ and relation $R$ resulting in a graph 
that belongs to a class of bounded local cliquewidth, but we know, by \cref{lem:super-coloring-intro}, that one of them does.
So how do we know that we will receive a correct answer in the required running time?

First, we use the fact that interpretations with bounded-range formulas preserve classes with bounded local cliquewidth (\Cref{lem:loc-bdcw}).
Second, we know that for every class $\DD$ with bounded local cliquewidth
there is a model checking  algorithm that is guaranteed to be efficient on graphs from $\DD$ only,
but yields correct answers for all graphs (see \cref{thm:mc-local-cw}).
By applying this algorithm in parallel we are therefore guaranteed to efficiently get a correct answer.
This completes the sketch of the proof of the main theorem, \cref{thm:mc_bounded_local_cliquewidth}, using the main lemma.
The details are presented in \cref{sec:main-proof}.

\paragraph*{Proof of main lemma.}
We now outline the proof of the main lemma.
See \cref{sec:main-lemma} for the complete argument.
We use the following fundamental result based on the 
$(p,q)$-theorem \cite{Matousek:2004:BVI:1005787.1005789}
(see \cref{thm:pq} below).

\begin{theorem}\label{thm:vc-duality}
  For every $d$ there is a number $k$ such that for every binary relation  $E\subset A\times B$ of VC-dimension at most $d$, one of two cases holds:
  \begin{itemize}
    \item there is a set $A'\subset A$ with $|A'|\le k$, such that for every $b\in B$ there is $a\in A'$ with $E(a,b)$, or 
    \item there is a set $B'\subset B$ with $|B'|\le k$, such that for every $a\in A$ there is $b\in B'$ with $\neg E(a,b)$.
  \end{itemize}
\end{theorem}

To prove \cref{lem:super-coloring-intro}, we proceed as follows. 
The starting point is again \cref{lem:gaifman-coloring}.
Let $r$ and $t$ be given by that lemma. 
Fix a graph $G\in \CC$ and its coloring as in \cref{lem:gaifman-coloring}.
Assume, for the sake of simplicity,
that every color class $C$ is either \emph{large},
that is contains $3$ vertices with mutual distance larger than $2r$,
or is \emph{small}, that is, contains a \emph{central} vertex $c_0\in C$ such that 
every vertex $v\in C$ is within distance at most $2r$ from $c_0$.
This assumption is without much loss of generality, as 
every class that is neither large nor small can be partitioned into two new classes 
that are both small.
We construct the set $S$ as follows:
\begin{itemize}
  \item for every large color class $C$, pick three elements 
   which are mutually at distance larger than $2r$,
  and add them to $S$,
  \item for every small color class $C$, pick a central vertex $c_0\in C$, and add it to $S$,
  \item for every pair $A,B$ of color classes, let $S_{AB}\subset A\cup B$ be the result of applying \cref{thm:vc-duality} to the binary relation  $E_\phi\subset A\times B$ 
   where $E_\phi=\setof{(a,b)\in A\times B}{G\models \phi(a,b)}$. Add $S_{AB}$ to $S$.
\end{itemize}
This completes the construction of $S$. Note that $|S|\le \Oof(t\cdot~k^2)$, where $k$ is given by \cref{thm:vc-duality}.
Correctness of the construction is verified for the radius $5r$. This amounts to proving that there are no vertices $u,v,u',v'\in V(G)$ such that:
\begin{itemize}
  \item $\dist(u,v)>5r$ and $\dist(u',v')>5r$,
  \item $\phi(u,S)=\phi(u',S)$ and $\phi(S,v)=\phi(S,v')$,
  \item $\phi(u,v)$ holds and $\neg\phi(u',v')$ holds.
\end{itemize}
Assuming that such vertices exist, a contradiction is reached
with the assumption that $\phi(u,v)$ depends only on the color of $u$ and the color of $v$ whenever $\dist^G(u,v)>r$.
This is done by performing a case analysis, depending on the sizes (large/small) of the color class $C(u')$ of $u'$ and the color class $C(v)$ of $v$.

We showcase one of the four cases: when $C(u')$ and $C(v)$ are both small.
As $\dist(u,v)>5r$ and  $C(v)$ is small, it follows that 
$\dist(u,w)>r$ for all $w\in C(v)$. Since $\phi(u,v)$ holds, it follows 
that $\phi(u,w)$ holds for all $w\in C(v)$. In particular, for $S(v)=S\cap C(v)$ we have
$\phi(u,S(v))=S(v)$. As $\phi(u,S)=\phi(u',S)$ it follows that $\phi(u',S(v))=S(v)$ as well. By a symmetric argument, using the fact that $C(u')$ is small, we get that $\phi(S(u'),v) = \emptyset$.
This contradicts the construction of the set $S_{AB}\subset S$
for the pair $A=C(u')$ and $B=C(v)$.

The case when one of $C(u')$ and $C(v)$ is small and the other one is large uses similar arguments. 
The case when both classes are large is even more elementary, as it does not invoke the construction of the sets $S_{AB}$, and only relies on the existence of the three-element scattered sets in each of $C(u')$ and $C(v)$, that where selected to $S$.

This finishes the sketch of the proof of the main lemma, and hence also of the main theorem.
Note that in Section~\ref{sec:main-lemma} we state a slightly stronger version of Lemma~\ref{lem:super-coloring-intro}, which is suited for treating 
interpretations in which the domain formula $\delta(x)$ is arbitrary.
In Section~\ref{sec:main-proof}, 
we prove  Theorem~\ref{thm:mc_bounded_local_cliquewidth}.

\section{Defining the relationship between far apart vertices}
\label{sec:main-lemma}

In this section we prove our main technical tool, \cref{lem:super-coloring}.
First we need some notation. For a formula $\phi(x,y)$, elements $u,v$ and a set $S$ of elements of a structure $G$, write:
\begin{align*}
    \phi(u,S)&:=\setof{s\in S}{G\models \phi(u,s)}\\
    \phi(S,v)&:=\setof{s\in S}{G\models\phi(s,v)}.
\end{align*}

\begin{lemma}[Main lemma]\label{lem:super-coloring}
  Let $\CC$ be a class of graphs and 
  let $\phi(x,y)$ be an \FO formula that has bounded VC-dimension on $\CC$. 
  Then there are numbers $s,r\in \N$ 
  such that for every graph $G \in \CCC$ and every $U \subset V(G)$ there is a 
set $S\subset U$ of size at most $s$
such that for any two vertices $u,v\in U$ with $\dist^G(u,v)>r$,
  whether or not $\phi(u,v)$ holds depends only on $\phi(u,S)$ and $\phi(S,v)$, where
\emph{depends only on} means that  $G \models \phi(u,v)$ iff $G \models \phi(u',v')$
 for any two pairs $u,v$ and $u',v'$ from $U$ 
satisfying the following condition:
\[
\begin{array}{l@{\qquad}l@{\qquad}c}
    \phi(u,S) =    \phi(u',S) & \dist^G(u,v)>r,\\[-5pt]
    && (\ast)\\[-5pt]
    \phi(S,v)=    \phi(S,v')&\dist^G(u',v')>r.
\end{array}
\]

\end{lemma}

The following property will play a key role in the proof of the main results of this section.

\begin{definition}\label{def:duality}
  Let $E\subset A\times B$ be a binary relation.
  Say that $E$ has a \emph{duality of order $k$} if at least one of two cases holds:
  \begin{itemize}
    \item[a)] there is a set $A'\subset A$ of size at most $k$ such that for every $b\in B$ there is some $a\in A'$ with $\neg E(a,b)$, or 
    \item[b)] there is a set $B'\subset B$ of size at most $k$ such that for every $a\in A$ there is some $b\in B'$ with $E(a,b)$.
  \end{itemize}
\end{definition}

A \emph{set system} $\cal F$ on a set $X$ is a family $\cal F$ of subsets of $X$. The \emph{VC-dimension} of $\cal F$ is the maximal size (or $+\infty$) of a subset $A\subset X$ such that $\setof{F\cap A}{F\in \cal F}=\mathcal{P}(A)$.
For $m\in \N$ let $\pi_{\cal F}(m)$  denote the \emph{shatter function} of $\cal F$, defined as 
\[
\pi_{\cal F}(m) :=  \max \Big\{ \big|\{F\cap A  \sth F\in \cal F\}\big| \sth A \subset X, |A| \le m \Big\},
\]
i.e.,~the  
maximum, over all sets $A\subset X$ with $|A|\le m$, of the cardinality of
$\setof{F\cap A}{F\in \cal F}$.
It is well known that if $\cal F$ has VC-dimension $d$ then $\pi_{\cal F}(m)=\Oof(m^d)$.

Define the VC-dimension of a binary relation $E\subset X\times Y$ as the VC-dimension of the set system 
$\setof{E(X,y)}{y\in Y}$~on~$Y$.

The following is a special case of the \emph{$(p,q)$-theorem}, stated below.
\begin{theorem}\label{thm:pq-special}
  For every $d\in \N$ there is some $k\in\N$ such that the following holds.
  Let $E\subset A\times B$ have VC-dimension at most~$d$, where $A$ and $B$ are finite. Then $E$ has a duality of~order~$k$.
\end{theorem}

This result follows from
the proof of the conjecture of Hadwiger and Debrunner, see 
Matou{\v s}ek~\cite[Theorem 4]{Matousek:2004:BVI:1005787.1005789}.
In the following formulation,
which is dual to the formulation 
of Matou{\v s}ek, the set system $\cal F$ is infinite.

\begin{theorem}[\cite{Matousek:2004:BVI:1005787.1005789}]
  \label{thm:pq}
	Let $\cal F$ be a set system on $U$ with $\pi_{\cal F}(m)=o(m^k)$,
	for some integer $k$, and let $p\ge k$.
	Then there is a constant $N$ such that the following holds for every finite set $V\subset U$: 
	if for every $V'\subset V$ with  $|V'|\le p$ there is some $F\in \cal F$ containing $V'$, then there is a family $\cal F'\subset \cal F$ with $|\cal F'|\le N$ and $V\subset \bigcup \cal F'$.
\end{theorem}
\begin{proof}[Proof of \cref{thm:pq-special}]
  Let $\cal F$ be the disjoint union of all finite set systems of VC-dimension at most $d$.
  Then $\cal F$ has VC-dimension at most $d$ as well, and therefore 
   $\pi_{\cal F}(m)=\Oof(m^d)=o(m^{d+1})$.
   Apply \cref{thm:pq} to $p=d+1$, obtaining a number $N$ with the following property: for every set system $\cal G$ on a finite set $V$ of VC-dimension at most $d$, such that every $p$ elements of $V$ are contained in some element of $\cal G$, there is a set of at most $N$ elements of $\cal G$ whose union contains $V$.

Let $E\subset A\times B$ have VC-dimension at most $d$, and let $\cal G=\setof{E(A,b)}{b\in B}$ be the corresponding set system~on~$A$.

Suppose there is a set $A'\subset A$ of size at most $p$ such that for every $b\in B$ there is some $a\in A'$ with $\neg E(a,b)$. Then $E$ has a duality of order $p=d+1$.

Otherwise, for every $A'\subset A$ of size at most $p$ there is some $b\in B$ such that $E(a,b)$ holds for all $a\in A'$.
This means that every subset of $A$ of size at most $p$ is contained in some element of $\cal G$. Hence, there is a 
subset $\cal G'\subset \cal G$ with $|\cal G'|\le N$ such that $A=\bigcup\cal G'$. This means that there is a
set $B'\subset B$ with $|B'|\le N$
such that for every $a\in A$,
$E(a,b)$ holds for some $b\in B'$.
Then $E$ has a duality of order $N$.

In either case, $E$ has a duality of order $\max(d+1,N)$.
  \end{proof}

Fix a partition $\cal P$ of 
a set $V$. For an element $v\in V$,
the \emph{class} of $v$, denoted $C(v)$,
is the unique $C\in \cal P$ containing~$v$.
In the context of the next theorem,
a \emph{pseudometric} is a symmetric function $f \from V \times V \to \R^+ \cup \{+\infty\}$ satisfying the triangle inequality.

\begin{theorem}\label{thm:independence-nip}
  Fix  $r,k,t\in\N$.
  Let $V$ be a finite set equipped with:
\begin{itemize}
  \item a binary relation $E\subset V\times V$
  such that for all $A\subset V$ and $B\subset V$, $E\cap (A\times B)$ has a duality of order $k$,
  \item a pseudometric $\dist\from V\times V\to \R_{\ge 0}\cup\set{+\infty}$,
  \item a partition $\cal P$ of $V$ with $|\cal P|\le t$,
\end{itemize} 
such that $E(u,v)$ depends only on $C(u)$ and $C(v)$ for all $u,v$ with $\dist(u,v)>r$.
Then there is a set $S\subset V$ of size $\Oof(k t^2)$ 
such that $E(u,v)$ depends only on $E(u,S)$ and $E(S,v)$ for all $u,v\in V$ with $\dist(u,v)>5r$.
\end{theorem}

\begin{proof}
  Say that a class $C\in\cal P$ is \emph{large}
  if there are $s_1,s_2,s_3\in C$ with 
  mutual distance larger than $2r$. 
  Say that a class $C\in\cal P$ is \emph{small}
  if there is $c_0\in C$ such that $\dist(c,c_0)\le 2r$ for all $c\in C$.
  If a class $C\in\cal P$ is not large then 
  there are $s_1,s_2\in C$ such that 
  $\dist(c,s_1)\le 2r$ or $\dist(c,s_2)\le 2r$ for all $c\in C$.
  For every class $C\in\cal P$ that is neither large nor small, pick arbitrarily any  such $s_1$ and $s_2$ 
  and let $C_1=\setof{c\in C}{c\neq s_2, \dist(c,s_1)\le 2r}$ and $C_2:=C\setminus C_1$.
  Thus, by splitting the class $C\in\cal P$ 
  into two  classes $C_1$ and $C_2$, we arrive at the situation where both $C_1$ and $C_2$ are small. 
  Hence, by at most doubling the number $t$ of classes,  
  we may assume that every class $C\in\cal P$ is either large or small.

\paragraph*{Construction of $\boldsymbol{\,S}$.}
We now construct the set $S$.
For every ordered pair $(C,D)\in \cal P^2$ of classes let $S_{CD}\subset C\cup D$ be a duality of order $k$ for $E\cap (C\times D)$, that is, $|S_{CD}|\le k$ and one of two cases holds:
\begin{itemize}
  \item for every $c\in C$ there is some $d\in S_{CD} \subseteq D$ with~$E(c,d)$, or 
  \item for every $d\in D$ there is some $c\in S_{CD} \subseteq C$ with~$\neg E(c,d)$.
\end{itemize}
Such a set $S_{CD}$ exists by the duality assumption of the lemma. Note that $S_{CD}$ and $S_{DC}$ are usually not the same and that we allow $C=D$ in the definition of $S_{CD}$.
Let $S\subset V$ be the set containing the following elements:
\begin{itemize}
  \item 
 for every class $C\in \cal P$ that is large,
any three elements $s_1,s_2,s_3\in C$ with mutual distance larger than $2r$,
  \item a center $c_0$ of every small class $C$, so that $\dist(w,c_0)\le 2r$ for all $w\in C$,
  \item all elements of $S_{CD}$, for every pair $(C,D)\in\cal P^2$.
\end{itemize}
Clearly, $S$ has $\Oof(kt^2)$ elements.

\paragraph*{Correctness.}
We show that $S$ satisfies the condition in the lemma.
Write $S(w)$ for $S\cap C(w)$, for $w\in V$. 
Towards a contradiction, suppose $u,v,u',v'\in V$ are such that:
\begin{enumerate}
  \item $E(u,v)$ and $\neg E(u',v')$,
  \item $\dist(u,v)>5r$ and $\dist(u',v')>5r$,
  \item $E(u,S)=E(u',S)$,
  \item $E(S,v)=E(S,v')$.
\end{enumerate}
We show that this yields a contradiction.

For a pair of classes $C,D\in \cal P$, possibly with $C=D$,
say that $E$ \emph{generically holds}
between $C$ and $D$ if $E(c,d)$ holds for some $c\in C$ and $d\in D$ 
such that $\dist(c,d)>r$.
Similarly define when $\neg E$ generically holds between $C$ and $D$.
Note that if $E$ generically holds between $C$ and $D$ then 
$E(c,d)$ holds for all $c\in C$ and $d\in D$ such that $\dist(c,d)>r$, by 
the assumption of the theorem. The same applies to $\neg E$.

By assumption, $E$ generically holds between $C(u)$ and $C(v)$,
whereas $\neg E$ generically holds between $C(u')$ and $C(v')$.

\begin{claim}\label{cl:large}
  The following hold:
  \begin{enumerate}
    \item If $C(u')$ is large, then
    $\neg E$ generically holds between $C(u')$ and $C(v)$.
    \item If $C(v)$ is large, then
    $E$ generically holds between $C(u')$ and $C(v)$.
  \end{enumerate}
\end{claim}
\begin{proof}
  We prove the first item, as the other one follows by symmetry. 
  The following situation is depicted in \Cref{fig:nip_claim1}.

Suppose $C(u')$ is large.
Then there are three elements in $S(u')$ with mutual distance larger than $2r$. At most one of them can be at distance at most $r$ from $v'$. So
we have $s_1,s_2\in S(u')$
with $\dist(s_i,v')>r$ for $i=1,2$. Then $\neg E(s_1,v')$ and 
$\neg E(s_2,v')$ holds since $\neg E$ generically holds between $C(u')$ and $C(v')$.
Since $s_1, s_2\in S$
and $E(S,v')=E(S,v)$, it follows that $\neg E(s_1, v)$ and 
$\neg E(s_2, v)$ hold as well.
As above, $v$ can be at distance at most $r$ only from one of $s_1$ and $s_2$.
It follows that $\neg E$ generically holds between $C(u')$ and $C(v)$.
\end{proof}

\begin{figure}[ht]
  \begin{center}
  \includegraphics[scale=0.75]{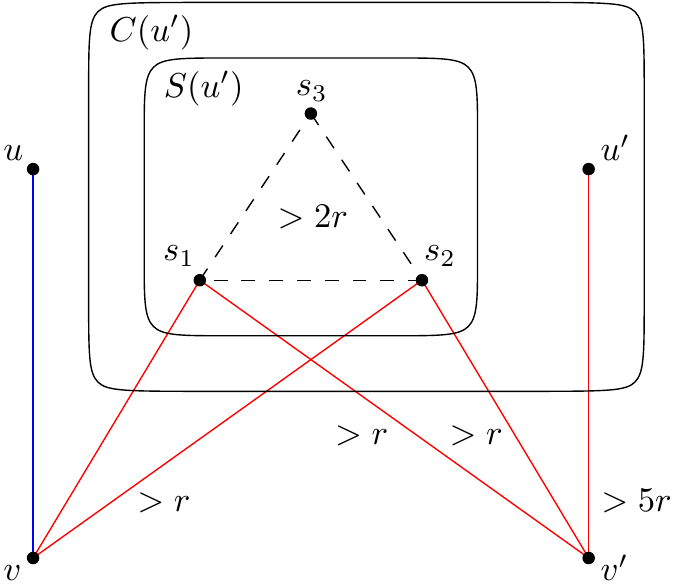}
  \end{center}
  \caption{
    A visualization of Claim \ref{cl:large}.
  Edges are annotated with the distances given by the pseudometric $\dist$.
  A \textcolor{blue}{blue} edge denotes an $E$ connection. 
  A \textcolor{red}{red} edge denotes a $\neg E$ connection.
  A \textcolor{gray}{dashed} edge is used when only the distance is of relevance.
  }
  \label{fig:nip_claim1}
\end{figure}

Consequently, $C(u')$ and $C(v)$ cannot both be large as 
it cannot be the case that simultaneously $E$ and $\neg E$ generically hold between them.


\medskip
We now show that we also arrive at a contradiction if both $C(u')$ and $C(v)$ are small. Later we will consider the case when one of them is small and the other one is large.

\begin{claim} \label{cl:star}The following hold:
  \begin{enumerate}
    \item If $C(u')$ is small, then 
    $\neg E(s,v)$ holds for all $s\in S(u')$.
    \item If $C(v)$ is small, then 
    $E(u',s)$ holds for all $s\in S(v)$.
  \end{enumerate}
\end{claim}
\begin{proof}Again we prove the first item, as 
  the other one follows by symmetry.
  The following situation is depicted in \Cref{fig:nip_claim2}.

  Fix $s\in S(u')$.
  Observe that $\dist(v',s)>r$.
Indeed, suppose $\dist(v',s)\le r$.
As $C(u')$ is small, $\dist(s,u')\le 4r$. Together this gives $\dist(v',u')\le 5r$, a contradiction. 

As $\neg E$ generically holds between $C(u')$ and $C(v')$, it follows that $\neg E(s,v')$ holds.
Since $E(S,v)=E(S,v')$, we get that $\neg E(s,v)$ holds.
\end{proof}

\begin{figure}[ht]
  \begin{center}
  \includegraphics[scale=0.75]{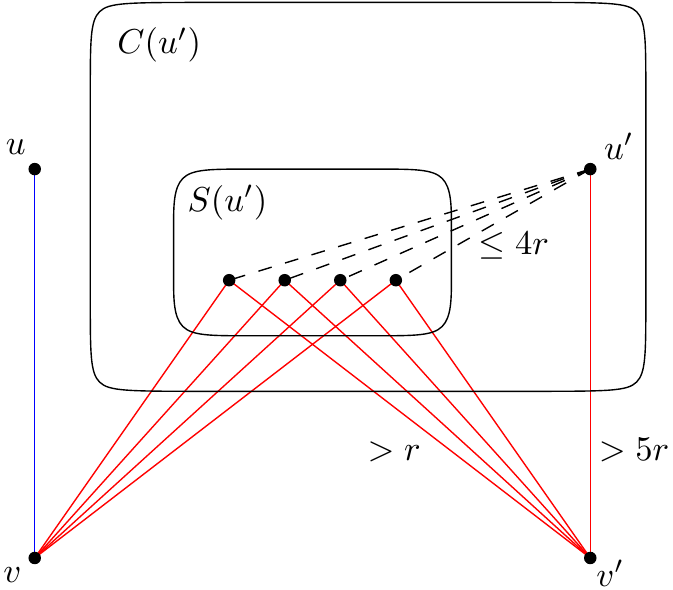}
  \end{center}
  \caption{A visualization of Claim \ref{cl:star}. The same notation as in Figure \ref{fig:nip_claim1} is used.}
  \label{fig:nip_claim2}
\end{figure}

Suppose both $C(u')$ and $C(v)$ are small. Then $E(u',s)$ holds for all
$s\in S(v)$,
and $\neg E(s,v)$ holds for all $s\in S(u')$,   contradicting the construction of $S_{C(v)C(u')}$,
as we have that either:
\begin{itemize}
  \item for every $b \in C(u')$ there is some $a \in S(v)$ with $\neg E(a,b)$, a contradiction to $u'$ being $E$-connected to every vertex from $S(v)$, or
  \item for every $a \in C(v)$ there is some $b \in S(u')$ with $E(a,b)$, a contradiction to $v$ being $E$-connected to no vertex from $S(u')$.
\end{itemize}

\medskip
So we are left with the case 
when exactly one of $C(u')$ and $C(v)$ is small.
By symmetry, we may assume that 
$C(u')$ is small: otherwise $C(v)$ is small and, up to  replacing $E(x,y)$ with $\neg E(y,x)$ and $u,v,u',v'$ with $v',u',v,u$, we are in the same case.

So $C(u')$ is small and $C(v)$ is large.
Then by \cref{cl:large},
 $E$ generically holds between $C(u')$ and $C(v)$.
 And by \cref{cl:star},
  $\neg E(s,v)$ holds for all $s\in S(u')$.

\begin{claim}\label{cl:star'}
  We have $\dist(u',v)\le 3r$.
\end{claim}
\begin{proof}
  Since $\neg E(s,v)$ holds for all $s\in S(u')$, in particular for the selected center $c_0\in S(u')$ of the small class $C(u')$ we have that $\neg E(c_0,v)$ holds.
   Since 
 $E$ generically holds between $C(u')$ and $C(v)$,
    it must be the case that ${\dist(c_0,v)\le r}$.
   Together with  $\dist(u',c_0)\le 2r$ this yields ${\dist(u',v)\le 3r}$.
\end{proof}


Since $\neg E(s,v)$ holds for all $s\in S(u')$,
 by construction of $S_{C(v)C(u')}$ 
 there is some $s\in S(v)$ such that 
 $\neg E(u',s)$ holds. 
  Then also $\neg E(u,s)$ holds, as $E(u,S)=E(u',S)$.
  As $E$ generically holds between $C(u')$ and $C(v)$, it follows that $\dist(u',s) \le r$.
   For a similar reason, $\dist(u,s)\le r$.
   Hence ${\dist(u,u')\le 2r}$.
   With \cref{cl:star'} this yields $\dist(u,v)\le 5r$, a contradiction.
  \end{proof}



\cref{lem:super-coloring} now follows from \cref{thm:independence-nip}.
\begin{proof}[Proof of \cref{lem:super-coloring}]
Let $\CC$ and $\phi$ be as in the assumptions of the lemma and let $d$ be the bound on the VC-dimension of $\phi$ on $\CC$. 
Let $G\in C$ and $U \subset V(G)$. 
By Corollary~\ref{lem:gaifman-coloring} we know that there exist numbers $r'$ and $t$ such that $V(G)$ can be colored by at most $t$ colors such that for all vertices $u,v$ of $G$ with $\dist^G(u,v)>r'$, $\psi(u,v)$ depends only on the colors of $u$ and $v$. 
Let $H = \interp I_{\phi}(G)[U]$ and for each color class $C_i \subset V(G)$ with $i \in[t]$ let $D_i$ be its restriction to the graph $H$, i.e. $D_i = C_i \cap U$.

We then have the following:
\begin{itemize}
\item For every $A \subset V(H)$ and $B \subset V(H)$ the VC-dimension of $(A\times B) \cap E(H)$ is bounded, and by Theorem~\ref{thm:pq-special} therefore $(A\times B) \cap E(H)$ has a duality of order $k$ depending only on $d$.
\item The function $\dist$ on $V(H)\times V(H)$ defined by setting $\dist(u,v) :=\dist_G(u,v)$ for every $u,v\in V(H)$ is a pseudometric.
\item  $\cal P = \set{D_i,\ldots, D_t}$ is a  partition $V(H)$ into sets such $uv \in E(H)$ depends only on the classes of $u$ and $v$ in $\cal P$ for every $u,v \in V(H)$ with $\dist(u,v)>r'$.
\end{itemize}
We can therefore apply Theorem~\ref{thm:independence-nip} to obtain a subset $S$ of $V(H)$ of size $\Oof(k t^2)$ 
such that $uv \in E(H)$ depends only on $E(u,S)$ and $E(S,v)$ for all $u,v\in V$ with $\dist(u,v)>5r'$.

Since $V(H) = U$ and $uv \in E(H)$ if and only if $G \models \phi(u,v)$, this concludes the proof after setting $r:=5r'$.
\end{proof}

\section{Model checking on interpretations of classes bounded local cliquewidth}
\label{sec:main-proof}

In this section we prove the main result of the paper.
Before we get started, we need to fix some notation.







\subsection{Graph classes}
We work with classes $\CC$ of graphs that are possibly equipped 
with unary predicates, constants, and \emph{flags}, that is, relation symbols of arity $0$ (a flag $f$ therefore evaluates to a Boolean $f_G\in\set {\textit{true},\textit{false}}$, for each structure $G\in \CC$).
More precisely, each class $\CC$ has a fixed finite signature $\Sigma$
which contains the binary relation symbol $E$, and relation symbols of arity $0$ or $1$, and constant symbols. Moreover, $E$ is interpreted as a symmetric, irreflexive relation in  each $G\in\CC$.
By abuse of language, we call structures in $\CC$ \emph{graphs}.
We will usually not mention the signature of a graph class explicitly, unless necessary.
We say that a class $\CC$ as above has \emph{bounded local cliquewidth}
if the class of underlying (usual) graphs has bounded local cliquewidth.

If $\Sigma$ and $\Gamma$ are two signatures with $\Sigma\subset\Gamma$, and $G$ is a $\Sigma$-structure,
then any $\Gamma$-structure~$G'$ obtained from $G$ by 
interpreting the symbols from $\Gamma$ not in $\Sigma$ 
is called a $\Gamma$-\emph{expansion} of~$G$.


\subsection{(Local) Cliquewidth}

We assume familiarity with the notions of treewidth and of cliquewidth. We denote the cliquewidth of a graph $G$ by $\cw(G)$.
We will need the following results. 

\begin{theorem}[\cite{ce12}]
\label{thm:cw_transductions}
Let $\CC$ be a class of graphs which is interpretable in a graph class of bounded cliquewidth. Then $\CC$ is of bounded cliquewidth.
\end{theorem}

\begin{theorem}[\cite{ce12}]
\label{thm:cw_mc}
There is a function $h\from\N\to\N$, a constant~$c$,
and an algorithm that, given a (colored) graph $G$ and a sentence $\phi \in \FO$
decides whether $G \models \phi$ in time 
$h(\cw(G)+|\phi|)\cdot |G|^c$. 
\end{theorem}
We will also need the localized variant of cliquewidth. 
If $G$ is a graph, $v \in V(G)$, and $r \ge 0$, then we denote by $N_r^{G}[v]$ the set of vertices in $G$ of distance at most $r$ from $v$. For $\bar v \in V(G)^l$ we define $N_r^{G}[\bar v] := \bigcup_{i=1}^l N^G_r[v_i]$.

\begin{definition}\label{def:local-tw}
Let $G$ be a graph. For $r \ge 0$ we define $$\lcw_r(G) := \max \{ \cw\big( G[N^G_r[v]]\big) \sth v \in V(G) \}.$$  We say that a class $\CCC$ of graphs has \emph{bounded local cliquewidth} if there is a function $f \sth \N \rightarrow \N$ such that $\lcw_r(G) \le f(r)$ for all $G \in \CCC$ and $r \ge 0$. 
\end{definition}

Classes of graphs of bounded local cliquewidth include all classes of bounded local treewidth (defined analogously) such as classes of graphs of bounded degree, the class of planar graphs or more generally classes of graphs embedded on a surface of fixed genus. On the other hand, the class of apex graphs, i.e.,~graphs $G$ which are planar after removal of a single vertex, does not have bounded local treewidth or cliquewidth.
Another classical example of classes of graphs of bounded local cliquewidth are map graphs.
The following result has its roots in the work of Frick and Grohe~\cite{frickg01} (see also~\cite{dawargk07,gro07}).

\begin{theorem}\label{thm:mc-local-cw}
    There is a function $h\from\N\to\N$, a constant $c$,
    and an algorithm that, given a (colored) graph $G$ and a sentence $\phi \in \FO$
    decides whether $G \models \phi$ in time 
    \[h(\lcw_{h(|\phi|)}(G)+|\phi|)\cdot |G|^c.\]    
\end{theorem}

The theorem follows from the model checking algorithm for bounded cliquewidth
as well as standard techniques, nicely presented in \cite[Theorem 4.5]{gro07}.
Note that this algorithm yields correct answers on all classes of graphs,
however it is only efficient on classes where the local cliquewidth is bounded.

We also need some notation related to Gaifman locality.
A first-order formula $\varphi(x_1, \ldots, x_l)$ is $r$-\emph{local}
if for every graph $G$ and $l$-tuple $\bar v \in V(G)^l$,
$
  G \models \varphi(\bar v) 
  \Leftrightarrow
  G[N_r^{G}[\bar v]] \models \varphi(\bar v).$

\begin{corollary}
  \label{cor:gaifman}
  For every formula $\psi(x,y)$ there exist numbers $r$ and $q$ such that the following holds: For every graph $G$ there exists an $r$-local formula $\psi_G(x,y)$ of quantifier rank at most $q$ such that for all vertices $u,v \in V(G)$ we have that 
  $G \models \psi_G(u,v)$ 
  if and only if
  $G \models \psi(u,v)$. 
\end{corollary}


We use Gaifman to prove that interpretations of bounded local cliquewidth have bounded local cliquewidth, too.

\begin{lemma}
\label{lem:loc-bdcw}
	Let $\CC$ be a class of graphs of bounded local cliquewidth and let $\trans I_{\phi,\delta}$ be an interpretation 
    such that there exists $d$ such that $\phi(x,y)$ has range at most $d$.
	Then $\trans I_{\phi,\delta}(\CC)$ is a class of graphs of bounded local cliquewidth.
\end{lemma}
\begin{proof}
	Our task is to prove that for every $H \in \trans   I_{\phi,\delta}(\CC)$, vertex $v_0 \in V(H)$ and every $r$ the graph  $H[N_r^H[v_0]]$ has cliquewidth bounded in terms of $r$. 
	
Let  $G \in \CC$ be such that $H = \trans I_{\phi,\delta}(G)$. This means that
$H = \trans I_{\phi}(G)[U]$, where $U = \setof{u \in V(G)}{G \models \delta(u)}$. In particular, $H$ is an induced subgraph of $\trans I_{\phi}(G)$.

 Let $r'$ and $\phi'(x,y)$ be the locality parameter and $r'$-local formula obtained by applying Gaifman's theorem in the form of \cref{cor:gaifman} to $\phi(x,y)$. Set $\ell:= rd+r'$ and $G_0 = G[N_{\ell}^G[v_0]]$. We will consider the graph $\trans I_{\phi'}( G_0)$ and show that \begin{itemize}
	\item $\trans I_{\phi'}( G_0)$ has cliquewidth bounded by a function of $r$ (here we consider $d$ and $r'$ to be fixed constants), and
	\item $H[N_r^H[v_0]]$ is an induced subgraph of $\trans I_{\phi'}(G_0)$.
	\end{itemize}
	The lemma then follows because cliquewidth is preserved by taking induced subgraphs.
	
	To show the first item we first note that $G_0$ has cliquewidth at most $g(\ell)$, where $g$ is the local cliquewidth bounding function for $\CC$. Let $\CC_{g(\ell)}$ be the class of all graphs of cliquewidth at most $g(\ell)$. 
We then have that  $\trans I_{\phi'}(G_0) \in \trans I_{\phi'}(\CC_{\ell})$,  The claim then follows by Theorem~\ref{thm:cw_transductions}.
%
	
	It remains to show that $H[N_r^H[v_0]]$ is an induced subgraph of $\trans I_{\phi'}(G_0)$, i.e., that $N_r^H[v_0] \subseteq V(G_0)$ and for all $u,v \in N_r^H[v_0]$ it holds that $uv \in E(H)$ if and only if $G_0 \models \phi'(u,v)$.
	Let $u,v$ be two vertices in $N_r^H[v_0]$.
	Since both $u$ and $v$ are at distance at most $r$ from $v_0$ in $H$, by our assumption on the range of $\phi(x,y)$ they are at distance at most $dr$ from $v_0$ in $G$.
	Indeed, $ab \in E(H)$ is equivalent to $G \models \phi(a,b)$, which implies, by assumption, $\dist^G(a,b) < d$.  
	This means that  $u,v \in N_{rd}^G[v_0]$ and so both $u$ and $v$ are in $V(G_0)$.  Moreover, every vertex at distance at most $r'$ from $u$ or $v$ in $G$ is at distance at most $rd+r'$ from $v$ in $G$, and so $N_{r'}^G[u] \cup N_{r'}^G[v] \subset N_{\ell}^G[v_0]$. Thus for the $r'$-local formula $\phi'(x,y)$ it holds that $G_0 \models \phi'(u,v)$ if and only if $G \models \phi'(u,v)$, and from Corollary~\ref{cor:gaifman} we know that $G \models \phi'(u,v)$ if and only if $G \models \phi(u,v)$. We therefore get $G_0 \models \phi'(u,v)$ if and only if $uv \in E(H)$, as desired.
\end{proof}

We will rely on the following theorem proved by Grohe and Tur\'an~{\cite[Lemma 22]{Grohe2003LearnabilityAD}.
\begin{theorem}\label{thm:blcwNIP}
Let $\CC$ be a class with bounded local cliquewidth.
Then $\CC$ is NIP.
\end{theorem}
\Cref{lem:loc-bdcw,thm:blcwNIP} both hold when $\CC$ is a class of graphs equipped with unary predicates, constants, and flags.

\subsection{Flips}

For a graph $G$ and $m \in \N$, an \emph{$m$-flip} is an operation determined by a partition of $V(G)$ into sets $V_1,\ldots, V_{m'}$ with $m' \le m$  and a symmetric binary relation $R$ on $[m]$. The resulting graph has the same vertex set as $G$ and its edge relation is obtained from $E(G)$ by complementing the edges between any $x \in V_i$, $y \in V_j$ such that $(i,j) \in R$. Note that it can be $i=j$.
We will call the output of an $m$-flip operation on a graph $G$ also an \emph{$m$-flip of $G$}.
Also note that flips are reversible, that is if $H$ is an $m$-flip of $G$, then $G$ is an $m$-flip of $H$. Let $S \subset V(G)$. We say that an $m$-flip is \emph{guarded} by $S$ if each of the sets $V_1,\ldots, V_{m'}$ is of the form $\setof{v \in V(G)}{N_G(v)\cap S = A}$ for some $A \subset S$. Note that in this case $m' \le 2^{|S|}$.

\begin{lemma}
\label{lem:flipping}
Let $\CC$ be a NIP class of graphs, $\trans I_{\phi,\delta}$ be an interpretation, and $\DD = \trans I_{\phi,\delta}(\CC)$. There exist $s,r\in \N$, a signature $\Gamma$ 
expanding the signature of graphs by constant symbols and relation symbols of arity $0$, 
 a formula $\psi(x,y)$ in the signature $\Gamma$ that has range at most $r$, such that  the following holds.
For every $H \in \DD$ there exists a graph $F(H)$ and a graph $\wh G(H)$ in the signature $\Gamma$, such that:
\begin{itemize}
\item $F(H)$ is a $2^s$-flip of $H$, guarded by a set $S \subset V(H)$ of size at most $s$,
\item $\wh G(H)$ is a $\Gamma$-expansion of some graph $G\in\CC$,
\item $F(H) =  \trans I_{\psi,\delta}(\wh G(H))$.
\end{itemize}
\end{lemma}

\begin{proof}
Since $\CC$ is NIP and therefore has bounded VC-dimen\-sion, we can apply Lemma~\ref{lem:super-coloring} to $\CC$ and $\phi$ to obtain numbers $r$ and $s$ with the properties claimed there.

Let $G \in C$ be such that $H = \trans I_{\phi,\delta}(G)$. Let $U = \setof{v \in V(G)}{G \models \delta(v)}$. 
Note that $U$ is the vertex set of $H$, meaning we have $H = \trans I_{\phi}(G)[U]$. 
For any $u,v \in V(H)$ it therefore holds that $G \models \phi(u,v)$ if and only if $uv \in E(H)$. 
Since $\phi(x,y)$ is symmetric, for any $u \in U$ and $X \subset U$ it holds that $\phi(u,X) = \phi(X,u) = N_H(u)\cap X$.
Lemma~\ref{lem:super-coloring} then states the following: there is a set $S \subset V(H)$ of size at most $s$ such that 
for all $u,v \in V(H)$ with $\dist^G(u,v) > r$, whether or not $uv \in E(H)$ holds depends only on $N_H(u)\cap S$ and $N_H(v) \cap S$. 

We now describe the flip $F(H)$ of $H$ from the statement of the lemma. 
Let $c_1,\ldots,c_s$ be an enumeration containing all the elements of $S$ 
(possibly with repetitions, and possibly also containing some elements of $V(H)\setminus S$). This exists, since $|S|\le s$, and we may assume that $H$ has at least one vertex, the other case being trivial.
For every set $A\subset\set{1,\ldots,s}$, we define $V_A = \setof{v \in V(H)}{N_H(v)\cap S =
\setof{c_a}{a\in A}}$. This determines a partition of $V(H)$ into exactly $2^s$ sets.
Define a relation $R\subset 2^{[s]}\times 2^{[s]}$ as follows: $(A,B) \in R$ if and only if
there are $u \in V_A$ and $v \in V_B$ such that $\dist^G(u,v) > r$ and $uv \in E(H)$. 
Let 
$F(H)$ be the $2^s$-flip of $H$ determined by $R$; it is guarded by $S$, and $|S| \le s$.
Thus, the first statement of the lemma is satisfied.

We next describe the graph $\wh G(H)$ and the formula $\psi$. The graph $\wh G(H)$ is obtained from $G$ by 
\begin{itemize}
  \item Marking the elements $c_1,\ldots,c_s$ using $s$ constant symbols, which we also denote $c_1,\ldots,c_s$.
  \item Encoding the relation $R$ using $2^{[s]}\times 2^{[s]}$ flags $f_{A,B}$, for $A,B\subset\set{1,\ldots,s}$. Namely, $f_{A,B}$ is set to true in $\wh G(H)$, for $A,B\subset\set{1,\ldots,s}$  if and only if the pair $(A,B)$ belongs to the relation $R$.
\end{itemize}

Note that since $G$ and $\wh G(H)$ have the same vertices and edges, the distances between vertices are the same in $G$ and $\wh G$.
Also note that since the signature of $G$ is a subset of the signature of $\wh G(H)$, the formulas $\delta(x)$ and $\phi(x,y)$ can be evaluated in $\wh G(H)$.
To define the formula $\psi$, first define the following formulas.
For $A\subset [s]$, let $P_A(x)$ be the formula \[P_A(x):=\delta(x)\land\bigwedge_{i\in A}{\phi(x,c_i)}\land \bigwedge_{i\in [s]\setminus A}{\neg \phi(x,c_i)},\] expressing that $x$ belongs to the part $V_A$.
Further, let
\[\alpha_R(x,y) = \bigvee_{(A,B)\subset 2^{[s]}\times 2^{[s]}}f_{A,B}\land(P_A(x) \land P_B(y))\] be the formula ``encoding'' the edges flipped according to $R$, so that $\wh G(H)\models\alpha_R(u,v)$ holds for $u,v\in U=V(H)$ such if and only if $(A,B)\in R$,
where $V_A$ is the part containing $u$ and $V_B$ is the part containing $v$.
In particular, by definition of $R$,
\begin{align}\label{br}
  \wh G(H)\models (\alpha_R(u,v)\leftrightarrow \phi(u,v))\qquad\text{for all $u,v \in U$ with $\dist^{G}(u,v) > r$}.
\end{align}
We now set $\psi(x,y) := (\phi(x,y) \oplus \alpha_R(x,y))\land (\dist(x,y)\le r)$
where $\oplus$ denotes the $\mathit{xor}$ operation,
and $\dist(x,y)\le r$ is the formula expressing the existence of a path of length at most $r$ from $x$ to $y$ in the underlying graph (here is $r$ is a fixed constant).
By construction, $\psi$ has range at most $r$.
Moreover, by \eqref{br} we have 
\begin{align}\label{qeq}
  \wh G(H)\models \psi(x,y)\leftrightarrow (\phi(x,y)\oplus \alpha_R(x,y)).  
\end{align}


It remains to argue that $F(H) = \trans I_{\psi, \delta}(\wh G)$. Clearly the vertex set of $\trans I_{\psi, \delta}(\wh G)$ is the same as $V(H)$, as both are equal to $U$. To see that they have the same edges, recall that $\alpha_R(x,y)$ encodes the flip which was used to obtain $F(H)$ from $H$, and the formula $\phi(x,y) \oplus \alpha_R(x,y)$ can be viewed in the exactly the same way -- it first introduces all edges of $H$ via $\phi(x,y)$ and then flips away exactly the edges $uv \in E(H)$ with $\dist^{G}(u,v) > r$. By \eqref{qeq}, this proves $F(H) = \trans I_{\psi, \delta}(\wh G)$,
as required.
\end{proof}

\subsection{Proof of the main theorem}
At last, we can prove Theorem~\ref{thm:mc_bounded_local_cliquewidth}, which we restate here for convenience.

\begin{customthm}{\ref{thm:mc_bounded_local_cliquewidth}}
\label{thm:main}
  The first-order model checking problem is fixed-parameter tractable on any class $\DD$ of graphs that is interpretable in a class $\CC$ of graphs of  bounded local cliquewidth. 
\end{customthm}

\begin{proof}
Let $\phi(x,y)$ and $\delta(x)$ be the formulas defining an interpretation such that $\DD \subset \trans I_{\phi,\delta}(\CC)$.
Since every class of graphs of  bounded local cliquewidth is NIP by \cref{thm:blcwNIP}, we can apply Lemma~\ref{lem:flipping} to $\CC$ and $\trans I_{\phi,\delta}$. This yields numbers $r$ and $s$, a formula $\psi(x,y)$ of 
range at most $r$, and for every graph $H \in \DD$ a graph $F(H)$
and an expansion $\wh G(H)$ of some graph $G\in \CC$,
with the properties claimed there. We first establish a crucial property of graphs $F(H)$.

\begin{claim}The class $\DD' := \setof{F(H)}{H \in \DD}$
 has bounded local cliquewidth. 
\end{claim}
\begin{proof*}
  By \cref{lem:flipping}, we have 
  $\DD' = \setof{\trans I_{\psi,\delta}(\wh G(H))}{H\in \DD}$.
  As the class $\setof{\wh G(H)}{H\in\DD}$ consists of expansions of graphs from $\CC$ by constant symbols and flags, it has has bounded local cliquewidth. As $\DD'$ is the image of this class under the interpretation of $\trans I_{\psi,\delta}$, the claim follows from 
  \cref{lem:loc-bdcw}.  
\end{proof*}

We now describe the model checking algorithm.
Assume we want to determine whether $H \models \rho$ for some
$H \in \DD$ and $\rho \in \FO$.
Let $\cal S = \setof{ (S,R)}{ S \subseteq V(G), |S| \le s, R \subseteq 2^S \times 2^S, \text{$R$ symmetric}}$.
Note that $|\cal S| = \Oof_s(|V(H)|^s)$ (meaning $|\cal S| \le c_s \cdot |V(H)|^s$ for some constant $c_s$ depending only on~$s$).
We will generate
a collection of graphs $H_{S,R}$, $(S,R) \in \cal S$, colored with at most $2^s$ colors, together with sentences $\rho_{S,R}$ of length $\Oof_{s}(|\rho|)$ such that 
$ H \models \rho \Leftrightarrow H_{S,R} \models \rho_{S,R},$
and for at least one ``correct'' choice $(S,R) \in \cal S$, 
$H_{S,R}$ is a vertex-coloring of $F(H)$, and therefore is a vertex-coloring of some graph in $\DD'$. This then implies the theorem as follows. 

Let $\wh \DD'$ be the class of all graphs from $\DD'$ equipped with  $2^s$ new unary predicates. 
Since at least one $H_{S,R}$ is a vertex-coloring of a graph in $\DD'$, at least one $H_{S,R}$ is contained in $\wh \DD'$.
We now run  the model checking algorithm from Theorem~\ref{thm:mc-local-cw} for all $(S,R) \in \cal S$ in parallel to determine whether $H_{S,R} \models \rho_{S,R}$,
and terminate the whole process once the first run stops. 
By Theorem~\ref{thm:mc-local-cw}, for any choice of $H_{S,R}$, $\rho_{S,R}$ this algorithm stops in time at most $h(\lcw_{h(|\rho_{S,R}|)}(H_{S,R})+|\rho_{S,R}|)\cdot |H_{S,R}|^c$.
For the ``correct'' choice $(S,R) \in \cal S$ (such that $H_{S,R}$ is a vertex-coloring of $F(H)$),
$\lcw_w(H_{S,R}) = \lcw_w(F(H))$ for all $w \in \N$, and since $F(H)$ comes from a class of bounded local cliquewidth, $\lcw_w(F(H)) \le f(w)$ for some function $f$ depending only on $\DD'$ 
(where $\DD'$ depends only on $\phi$, $\delta$ and $\DD$).
Moreover, $|\rho_{S,R}| \le \Oof_s(|\rho|)$, so in total for this run the algorithm stops after at most $h(f(h(\Oof_s(|\rho|)))+|\rho_{S,R}|)\cdot |H|^c$ steps. Because there are $\Oof_s(|V(H)|^s)$ of runs executed in parallel, this bounds the run time of our algorithm by 
$\Oof_{s,|\rho|} (|H|^{c+s})$. Since $s$ depends on $\DD$ but not on $\rho$, we obtain FPT run time as desired. 


We now describe the construction of the graphs $H_{S,R}$ and sentences $\rho_{S,R}$.
For each set $S\subset V(H)$ of size at most $s$ and 
 symmetric binary relation $R$ over the color set $2^S$,
 define a colored graph $H_{S,R}$ obtained from $H$ as follows.
First, 
 color each vertex $u \in V(H)$ by $\lambda_S(u) := N_H(u) \cap S$.
 Then, flip the adjacency (that is, an edge becomes a non-edge and vice versa) between every pair $u,v \in V(H)$ if and only if $(\lambda_S(u),\lambda_S(v)) \in R$. Since we go through all subsets $S$ of $V(H)$ of size at most $s$ and all possible flips guarded by $S$, for some choice of $S$ and $R$, the graph $H_{S,R}$ is a coloring of $F(H)$.
  
  To describe $\rho_{S,R}$, first let us consider the formula $\zeta_{S,R}$ which flips the edges of $H_{S,F}$ back to obtain $H$, given by
  $$ \zeta_{S,R}(x,y) := x\not= y  \land \bigvee_{(\alpha, \beta) \in R} (c_{\alpha}(x) \land c_{\beta}(y)), $$
  where $\alpha, \beta \in 2^S$ and the meaning of the predicate $c_{\alpha}(u)$ is that $\lambda_S(u) = \alpha$. 
  Then $\rho_{S,R}$ is obtained from $\rho$ by
  replacing each occurrence of $E(x,y)$ by 
   $E(x,y) \oplus \zeta_{S,R}(x,y)$. We have $H \models \rho \Leftrightarrow H_{S,R} \models \rho_{S,R}$.
\end{proof}

\section{Discussion}\label{sec:discussion}
We now discuss how our results fit into a broader picture.
For the purpose of this discussion, it is slightly more convenient to replace interpretations with the more general \emph{transductions}, which are defined below.
\emph{Simple non-copying transductions} (here, \emph{transductions} for short) are defined similarly as interpretations. 
First, they may nondeterministically color 
the input graph $G$ with a fixed number of colors, and afterwards they apply 
a fixed interpretation to the obtained colored graph, yielding an output graph.
Thus, a transduction maps a single graph $G$ to a set of possible output graphs,
where the various possible outputs correspond to the various possible colorings.
Say that a class $\CC$ \emph{transduces} a class~$\DD$,
or that $\DD$ \emph{transduces in} $\CC$,
if there is a transduction $T$ such that $\DD\subset T(\CC)$.
As in the case of interpretations, this defines a transitive relation. If the interpretation applied by the transduction $T$ comes from a logic $\LLL$, we say that $T$ is an $\LLL$-transduction.

In the previous section, replacing interpretations with  transductions would not make a difference in most places. In particular, a class 
$\CC$ is a transduction of a class with bounded local cliquewidth if and only if it is an interpretation of such a class, so our main result also holds for the more general notion. Such a replacement is not neutral in all contexts, however.

For a property $\cal P$ of graph classes, 
we say a class $\DD$ of graphs has \emph{structurally} $\cal P$, if $\DD$ \FO-transduces in some class~$\CC$ with property $\cal P$.
So for example, a class $\CC$ is structurally nowhere dense if it transduces (equivalently, interprets) in some 
nowhere dense class $\DD$.
Classes that transduce (equivalently, interpret) in a class with bounded local cliquewidth are exactly classes with structurally bounded local cliquewidth, and our main result concerns those classes.

A reformulation of \cref{qstn:transductions_preserve_tractability}, generalized to transductions instead of interpretations, therefore asks:
are structurally tractable classes tractable?
Let us evaluate the status of this question by listing classes~$\CC$ 
that are known to admit an {\FPT} algorithm for {\FO} model checking
and discussing what can be said about transductions $\DD$ thereof.
See \Cref{fig:dense_inclusion_diagramm} for an overview.

Particularly interesting cornerstones in this context are \emph{transduction ideals}.
We use this term to denote properties of hereditary graph classes,
that are preserved by (first-order) transductions.
By transitivity of the transduction relation,
for every property $\cal P$ of graph classes, the property ``structurally~$\cal P$'' forms a transduction ideal.

\Cref{qstn:transductions_preserve_tractability} suggests the existence of an {\FPT} algorithm for {\FO} model checking on structurally nowhere dense classes.
Up to now, this has been only confirmed for classes of structurally bounded degree,
and for classes of bounded shrubdepth (that is, transductions of classes of trees of bounded depth).
Our main result in particular implies that the same holds
for every class with structurally bounded local treewidth.

Besides classes of structurally bounded degree and classes of bounded shrubdepth, classes with structurally
boun\-ded local treewidth include structurally planar classes, classes with
structurally bounded genus, and structurally apex-minor-free graph classes.

A next step would be to consider 
classes with structurally bounded expansion, which are strictly weaker than structurally nowhere dense classes. They do not 
include all classes with structurally  bounded local treewidth,
however (see \Cref{fig:sparse_inclusion_diagramm}).

We note that the proof of tractability of nowhere dense classes~\cite{gks17}
is based on an iterative application of locality arguments,
combined with structural properties of nowhere dense classes.
In particular, every $r$-ball in a graph from a nowhere dense class $\CC$  
belongs to a nowhere dense class that is simpler in some sense, as is formalized by the notion of splitter games~\cite{gks17}.
Therefore, it is conceivable that an extension of our methods 
will allow to approach the problem of tractability of structurally nowhere  
dense classes.

\paragraph*{Monadically stable classes.}
Structurally nowhere dense classes are further generalized by \emph{monadically stable} classes.
A~class is monadically stable if it does not transduce\footnote{This is one of two places where the distinction between transductions and interpretations matters. For example, consider the class $\CC$ of graphs that can be obtained from a clique by placing a vertex in the middle of every edge. Then $\CC$ does not interpret the class of half-graphs,
but $\CC$ transduces the class of all graphs, since we can first color some subset of the middle vertices, and in this way encode any graph. The 
other place where the distinction matters is in the definition of monadically NIP classes, and the same example illustrates the issue.} the class of all half-graphs. 
In other words, monadically stable classes form the largest transduction ideal that does not contain the class of half-graphs.
Monadically stable classes 
were introduced by Baldwin and Shelah \cite{baldwin1985second},
and are a special case of \emph{stable classes}, which are one of the central objects of interest in stability theory. Stability theory is now the main focus of model theory.

There are strong connections between stability theory and (structurally) sparse graph classes.
Most notably, it was shown by Podewski and Ziegler~\cite{podewski1978stable} in the late 70's, long before the development of sparsity theory, that all nowhere dense classes (called \emph{superflat} in their paper) are monadically stable. 

The result of Podewski and Ziegler,
connecting sparsity theory with stability theory, has been brought to the attention of the sparsity community by Adler and Adler~\cite{AA14}, who observed that nowhere dense classes are the same as superflat classes.
By the result of Podewski and Ziegler,  nowhere dense classes, and therefore also structurally nowhere dense classes, are monadically stable.
In the other direction, it is conjectured~\cite{rw_stable} that a graph class is monadically stable if and only if it is structurally nowhere dense.

\paragraph*{Unstable classes.}
Classes with bounded cliquewidth are not necessarily monadically stable,
as the class of all half-graphs has bounded cliquewidth and is not monadically stable by definition.
On the other hand, the class of planar graphs is nowhere dense, and hence monadically stable, but has unbounded cliquewidth. 
Bounded cliquewidth is therefore incomparable to monadically stable (or nowhere dense) classes.
Nevertheless, \FO (and even $\MSO_1$) model checking is \FPT on these classes.
Bounded cliquewidth also forms a transduction ideal (even for \MSO transductions).
Thus, if $\CC$ has bounded cliquewidth and $\DD$ is a transduction thereof, then $\DD$ also has bounded cliquewidth.

Classes with bounded local cliquewidth are tractable (see Theorem \ref{thm:loc-cw}), but 
do not form a transduction ideal, as they are not closed under edge-complementation.
By our main result, all classes with structurally bounded local cliquewidth are also tractable, and those do form a transduction ideal.

Recently, Bonnet, Kim, Thomass\'e and Watrigant~\cite{twin-width1} introduced the notion of
\emph{twinwidth} and showed that classes of bounded twinwidth are preserved by
\FO transductions. 
So bounded twinwidth is a transduction ideal, and subsumes bounded cliquewidth, but is incomparable to structurally  bounded local cliquewidth\footnote{The class of cubic graphs has bounded local cliquewidth,
but unbounded twinwidth \cite{BonnetGKTW21}. On the other hand, 
consider the class $\CC$ of graphs $G$ such that each connected component of $G$ is a grid with an added apex vertex. Then $\CC$  has bounded twinwidth, but does not have structurally bounded local cliquewidth.
}.
Moreover, \FO model checking is FPT on classes with bounded twinwidth, but only 
assuming an appropriate decomposition is given as additional input~\cite{twin-width1}. 

Let us stress that our algorithm captures all known transduction ideals for which the model checking problem is FPT (without an additional decomposition given as input).

\paragraph*{(Monadically) NIP classes.}
The following notion, encompassing all the graph classes mentioned above, again originates in stability theory -- despite its name, stability theory does not only concern stable classes.
A class $\CC$ is \emph{monadically NIP}, or \emph{monadically dependent}, if it does not transduce the class of all graphs.
In other words, monadically NIP classes constitute the largest transduction ideal, apart from the one that contains all classes.
All the aforementioned graph classes are monadically NIP:
nowhere dense classes, classes of structurally bounded local cliquewidth,
classes of bounded twinwidth, etc.
It is conjectured~\cite[Conjecture 8.2]{gajarsky_2020} that \FO model checking is FPT on all {monadically NIP} classes.
Every monadically NIP class $\CC$ is in particular NIP,
that is, every formula $\phi(x,y)$ has bounded VC-dimension on~$\CC$.
Hence, \cref{lem:super-coloring-intro} applies to all such classes.

\medskip
With essentially the same proof as 
for \cref{thm:mc_bounded_local_cliquewidth}, we can obtain the following, more general statement.
We say the \FO model checking problem is \emph{conservatively FPT} on a class of structures $\CC$
if there is an algorithm that, for every \FO formula $\phi$ and structure $G\in \CC$, 
 decides whether $G \models \phi$,
and runs in time $f(\phi)\cdot n^c$ for every $G \in \CC$ with $n$ elements.
In contrast, an FPT model checking algorithm on $\CC$ is not required to give correct answers for structures outside $\CC$.
All the FPT \FO model checking algorithms we discussed so far are also conservatively FPT.

Let $\Sigma$ be the signature of graphs. A formula $\phi(x,y)$ has \emph{bounded range} if it has range $\le r$, for some $r\in \N$.
An interpretation $\interp I\from \Sigma\to\Gamma$,
where $\Gamma$ is a signature consisting of unary and binary relation symbols, has \emph{bounded range} if for all binary symbols $R\in\Gamma$,
the formula $\phi_R(x,y)$ has bounded range.
An interpretation \emph{with parameters} $\interp I\from \Sigma\to\Gamma$ is an interpretation $\interp I'\from \Sigma'\to \Gamma$, where $\Sigma'$ expands $\Sigma$ with constant symbols
and relation symbols of arity $0$ (flags).
For such an interpretation and class $\CC$ of $\Sigma$-structures, write $\interp I(\CC)$ 
for the class of all structures of the form $\interp I'(\str A')$,
where $\str A'$ is a $\Sigma'$-structure expanding some structure $\str A\in \CC$, by providing an interpretation of each constant symbol in $\Sigma'$ and not in $\Sigma$. A~class of \emph{colored graphs} is a class 
of structures over a signature $\Gamma=\set{E,U_1,\ldots,U_k}$, where $U_1,\ldots,U_k$ are unary relation symbols, and $E$ is interpreted as a binary symmetric, irreflexive relation. 

\begin{theorem}\label{thm:local_iff_global}
    Let $\CC$ be an NIP class of graphs.
    The first-order model checking problem is conservatively FPT on graph classes that interpret in $\CC$, if it is conservatively FPT on classes of colored graphs that interpret in~$\CC$ via a bounded-range interpretation with parameters.
\end{theorem}

Bounded-range interpretations (with parameters) of classes with bounded local cliquewidth 
again have bounded local cliquewidth and therefore a conservative FPT model checking algorithm.
Therefore, \Cref{thm:mc_bounded_local_cliquewidth} merely describes part of a bigger picture painted by \Cref{thm:local_iff_global}.
We believe it will serve as a crucial tool towards answering \cref{qstn:transductions_preserve_tractability} in other cases. 

\bibliographystyle{plain}



\end{document}